\documentclass[11pt,a4paper]{article}
\usepackage[T1]{fontenc}
\usepackage[utf8]{inputenc}
\usepackage{authblk} 

\usepackage{times}
\usepackage{helvet}
\usepackage{courier}
\usepackage{amssymb}
\usepackage{amsthm}
\usepackage{verbatim} 
\usepackage{enumerate} 
\usepackage{graphicx}
\usepackage{wrapfig} 
\newtheorem{definition}{Definition}[section]
\newtheorem{theorem}{Theorem}[section]
\newtheorem{lemma}{Lemma}[section]
\newtheorem{corollary}{Corollary}[section]
\newtheorem{proposition}{Proposition}[section]

\title{Canonical Logic Programs are Succinctly Incomparable with Propositional Formulas\thanks{Extended version of a paper with the same name in  KR2014.}}
\author{Yuping Shen}
\author{Xishun Zhao\thanks{Corresponding Author.}}
\affil{Institute of Logic and Cognition,\\
Department of Philosophy,\\
 Sun Yat-sen University,\\
510275 Guangzhou, P.R. China\\
\texttt{\{shyping, hsszxs\}@mail.sysu.edu.cn}}

\begin{document}
\maketitle

\begin{abstract}
\emph{Canonical (logic) programs} (CP) refer to normal logic programs  augmented with connective $not\ not$. In this paper we address the question of
 whether CP  are \emph{succinctly incomparable} with \emph{propositional formulas} (PF). Our main result shows that the PARITY problem, which
 can be polynomially represented in PF  but \emph{only} has exponential representations in CP. In other words, PARITY \emph{separates} PF from CP. Simply speaking, this means
 that exponential size blowup is generally inevitable  when  translating a set of formulas in PF into an equivalent program in CP (without introducing new variables). Furthermore,  since it has been shown by   Lifschitz and  Razborov  that there is also a problem that separates CP from PF (assuming $\mathsf{P}\nsubseteq \mathsf{NC^1/poly}$), it follows that CP and PF are
  indeed succinctly incomparable. From the view of the theory of computation,  the above result may also be considered as the separation of two \emph{models of computation}, i.e., we identify a language in $\mathsf{NC^1/poly}$ which is not in the set of languages computable by polynomial size CP programs.
\end{abstract}

\section{Introduction}
The  study of logic programs under answer set semantics, i.e., \emph{answer set programming} (ASP) \cite{Gelfond88thestable,lif08,DBLP:journals/cacm/BrewkaET11}, has been an active area in artificial intelligence  since the past decades. As a competing approach to SAT \cite{DBLP:series/faia/2009-185}, ASP has been successfully applied in many fields like Planning, Commonsense Reasoning, Scheduling, etc.

The relationship between logic programs and \emph{propositional formulas} (PF) gains a lot of attention in the literature. A well-known  theorem shown  by Lin \& Zhao \cite{Lin2004115} gives a method for translating a \emph{normal (logic) program} (LP) to a (logically) equivalent set of formulas in PF, without introducing additional variables. However, it has been observed that the translation may result in an exponential number of so-called \emph{loop formulas} in the worst case. In 2006, Lifschitz and Razborov proved that  such exponential  blowup is generally inevitable,  more precisely, they showed that (a variant of) the $\mathsf{P}$-complete \emph{problem} PATH has polynomial size representations in LP, however, it \emph{cannot} be  polynomially  represented in PF (assuming  $\mathsf{P}\nsubseteq \mathsf{NC^1/poly}$) \cite{Lifschitz:2006:WSM:1131313.1131316}. In other words, we say PATH \emph{separates} LP from PF.

As noted in \cite{Lifschitz:2006:WSM:1131313.1131316}, PF can be considered as  a special case of  \emph{(nondisjunctive) nested programs} (NLP) \cite{Lifschitz99nestedexpressions}, which is a  general  form  of programs that subsumes LP and some other kinds of programs. Therefore, NLP is \emph{stronger} than PF in terms of the \emph{succinctness} criterion (or the \emph{``comparative linguistics'' approach}) proposed in \cite{Gogic95thecomparative}:
\begin{quote}
That is, we consider formalism $A$ to be  stronger than formalism $B$ if and only if any knowledge base (KB) in $B$ has an equivalent KB in $A$ that is only polynomially longer, while
there is a KB in $A$ that can be translated to $B$ only with an exponential blowup.
\end{quote}
So the following footnote in \cite{lif08} seems convincing at first glance:
\begin{quote}
...ASP appears to be stronger than SAT in the sense of the ``comparative linguistics'' approach to knowledge
representation...
\end{quote}

However, since ASP involves many kinds of programs, the above statement probably needs  further clarification. Particularly, the so-called \emph{(nondisjunctive) canonical  programs} (CP)\footnote{Extends LP with connective $not\ not$.} \cite{Lee:2005:MCL:1642293.1642374,Lifschitz99nestedexpressions,conf/ijcai/LeeLY13}, is a ``minimal'' form of ASP that is equally expressive as PF, but looks more likely  \emph{not} succinctly stronger.  So a question naturally arises: \emph{Does there exist a problem  that separates CP from PF?} If there is such a problem, then CP and PF are  \emph{succinctly incomparable} (assuming $\mathsf{P}\nsubseteq \mathsf{NC^1/poly}$).

In this paper we address the question and give a \emph{positive} answer.  Our main result shows that the problem  PARITY  separates PF from CP. Simply speaking, this means an exponential size blowup is generally inevitable  when  translating a set of formulas in PF into an equivalent program in CP (without introducing new variables). The PARITY problem asks whether a binary string contains an odd number of 1's, and it is well-known  that (i) PARITY$\in\mathsf{NC^1/poly}$, i.e., it has polynomial representations in PF\footnote{$\mathsf{NC^1/poly}$ (or \emph{non-uniform} $\mathsf{NC^1}$) \emph{exactly} contains languages computable (i.e., representable) by polynomial size propositional formulas.} \cite{ccama,DBLP:books/daglib/0028687}, (ii) PARITY$\notin\mathsf{AC^0}$, i.e., it cannot be represented by \emph{polynomial} size boolean circuits with \emph{constant depth} and \emph{unbounded fan-in} \cite{DBLP:journals/mst/FurstSS84,Hastad:1986:AOL:12130.12132}.

To show PARITY separates PF from CP, we provide a procedure that simplifies  every PARITY  program $\Pi$ into a shorter, \emph{loop-free} program $\Pi'$. By Lin-Zhao Theorem (or the (generalized) \emph{Fages Theorem} \cite{fages94,Erdem:2003:TLP:986809.986814,You:2003:EAS:1630659.1630783}), $\Pi'$ is equivalent to its \emph{completion} $Comp(\Pi')$, the latter is essentially a constant depth, unbounded fan-in  circuit whose size is  polynomially bounded by $|\Pi'|$. According to   PARITY$\notin\mathsf{AC^0}$, these circuits must be of exponential size, consequently, there are no polynomial size CP  programs for PARITY.

From the view of the theory of computation,  the above result may also be considered as the separation of two \emph{models of computation} \cite{savage1998models}, i.e., we identify a language in $\mathsf{NC^1/poly}$ which is not in the set of languages computable by polynomial size CP programs. Based on the observation,  we point out more separation results on some classes of logic programs, e.g.,  PARITY separates logic programs with \emph{cardinality constraints} and \emph{choice rules} (CC) \cite{DBLP:journals/ai/SimonsNS02} from CP; assuming $\mathsf{P}\nsubseteq \mathsf{NC^1/poly}$,  CP and \emph{definite causal theories} (DT) \cite{mccain97thesis,Giunchiglia04nonmonotoniccausal} are succinctly incomparable;  \emph{two-valued} programs (TV) \cite{DBLP:conf/iclp/Lifschitz12} are strictly more succinct than CP and DT, etc.

The rest of the paper is organized as follows: Section \ref{sec:background} gives preliminaries to the semantics of canonical programs, the concepts of succinctness and the PARITY problem. In Section \ref{sec:boolean-completion} we briefly review the notation of boolean circuit, the completion semantics and the  Lin-Zhao theorem. Section \ref{sec:gen-simply-parity} illustrates how to simply an arbitrary PARITY program to be  loop-free and presents the main theorem. In Section \ref{sec:more-suc-result} we  discuss the importance of succinctness research and point out more results on a family of logic program classes. Conclusions are  drawn  in the last Section.

\section{Background}\label{sec:background}
\subsection{Canonical Programs}
The following notations are adopted from \cite{Lifschitz99nestedexpressions,Lee:2005:MCL:1642293.1642374}.
A \emph{rule element} $e$ is defined as $$e:= \top\ |\ \bot\ |\ x\ |\ not\ x\ |\ not\ not\ x$$ in which $\top,\bot$ are  $0$-ary connectives, $x$ is a \emph{(boolean) variable} (or an \emph{atom})  and $not$ is a unary connective\footnote{According to  \cite{Lifschitz99nestedexpressions},  $not\ not\ not\ x$ can be  replaced by $not\ x$.}.  A \emph{(nondisjunctive canonical) rule} is an expression of the form
\begin{equation}\label{SNE:rule}
H\leftarrow \ B
\end{equation}
where the \emph{head} $H$ is either a variable or the connective $\bot$, and the \emph{body} $B$ is a  finite set of rule elements.
 A \emph{canonical  program} (CP) $\Pi$ is a finite set of  rules, $\Pi$ is \emph{normal} if it contains no  connectives $not\ not$. A normal program $\Pi$ is \emph{basic}  if it contains no connectives $not$.  The following is a canonical program:
\begin{equation}\label{SNE:SNE1}
\begin{array}{c}
x_1 \leftarrow   not\ not\ x_1,\\
x_2 \leftarrow   not\ not\ x_2,\\
\end{array}
\begin{array}{l}
x_3 \leftarrow  not\ x_1, not\ x_2,\\
x_3  \leftarrow   x_1,\ x_2.\\
\end{array}
\end{equation}

The \emph{satisfaction relation} $\models$ between a set of variables $I$ and a rule element is defined as follows:
\begin{itemize}
\item $I\models \top$ and  $I\nvDash\bot$,
\item $I\models x$   iff $I\models not\ not\ x$ iff $x\in I$,
\item $I\models not\ x$ iff $x\notin I$.
\end{itemize}

Say  $I$ satisfies a set of rule elements $B$ if $I$ satisfies each rule elements in $B$. We say $I$ is \emph{closed} under a program $\Pi$, if $I$ is closed under every rule in $\Pi$, i.e.,
for each rule $H\leftarrow B\in\Pi$, $I\models H$ whenever $I\models B$. Let $\Pi$ be a basic program, $Cn(\Pi)$ denotes the  \emph{minimal} set (in terms of inclusion)  closed under  $\Pi$, we say $I$ is an \emph{answer set} of $\Pi$ if $I=Cn(\Pi)$.   Note that a  basic program has exactly one answer set.

The \emph{reduct} $\Pi^I$ of a canonical program $\Pi$ w.r.t. $I$ is a set of rules obtained from $\Pi$ via: (i) Replacing  each $not\ not\ x$  with $\top$ if $I\models x$, and with  $\bot$ otherwise; (ii) Replacing  each $not\ x$   with $\top$ if $I\nvDash x$, and with  $\bot$ otherwise. Observe that $\Pi^I$ must be a basic program. We say $I$ is an answer set of $\Pi$ if $I=Cn(\Pi^I)$, i.e., $I$ is an answer set of $\Pi^I$.

The following single rule canonical program $\Pi$:
\begin{equation}\label{KR:SNE-3}
 x  \leftarrow  not\ not\ x
\end{equation}
has two answer sets $\{x\}$ and $\emptyset$. To see this, check that $\Pi^{\{x\}}$ is  $\{x\leftarrow \top \}$, whose only answer set is $\{x\}$. Similarly,
$\Pi^{\emptyset}$ is $\{x\leftarrow \bot \}$, whose only answer set is $\emptyset$. For convenience, $\top$ in the body is often omitted.

For a set of rule elements $B$, define $var(B)=\{e\in B: e\mbox{ is a variable}\}$. E.g., $var(\{x_1, not\ x_2, not\ not\ x_3\})=\{x_1\}$.
Let $\Pi$ be a  program, by $var(\Pi)$ we denote the set of all variables involved in $\Pi$ and by $Ans(\Pi)$ we denote the set of all answer sets
of $\Pi$. E.g., let $\Pi$ be program (\ref{SNE:SNE1}), then $var(\Pi)=\{x_1,x_2,x_3\}$ and $Ans(\Pi)=\{\{x_1,x_2,x_3\},\{x_1\},\{x_2\},\{x_3\}\}$.
As a convention, by $\Pi_n$ we refer to a  program with $n$ variables  $\{x_1,\ldots,x_n\}$, i.e., $var(\Pi_n)=\{x_1,\ldots,x_n\}$. The \emph{size}  $|\Pi_n|$ of a program   $\Pi_n$,  is the number of rules in it.

\subsection{Problem Representation and Succinctness}
A \emph{string} is a finite sequence of \emph{bits} from $\{0,1\}$.    A string $w$ of length $n$ (i.e., $w\in\{0,1\}^n$) can be written as
 $w_1w_2\ldots w_n$, in which each bit  $w_i\in \{0,1\}$.
   Note that a string $w\in\{0,1\}^n$ defines a subset  of  variables $\{x_1,\ldots,x_n\}$, e.g., $1010$ stands for $\{x_1,x_3\}$. So a set of variables and a string is  regarded as the same. A \emph{problem}  (or \emph{language} ) $L$ is a set of strings.
\begin{definition}[Problem Representation]\label{def:LPcomp}
A problem $L$ can be \emph{represented}  in a class of  programs $\mathcal{C}$ (i.e., $L\in \mathcal{C}$), if  there exists a \emph{sequence}
of programs $\{\Pi_n\}$ ($n=1,2,\dots$) in $\mathcal{C}$  that \emph{computes} $L$, i.e.,  for every string $w\in\{0,1\}^n$,
$$w\in L \Leftrightarrow w\in Ans(\Pi_n).$$
 Moreover, say $L$ has polynomial  representations in $\mathcal{C}$ (i.e., $L\in \mbox{Poly-}\mathcal{C}$), if $L\in \mathcal{C}$ and $|\Pi_n|$ is bounded by a polynomial $p(n)$.
\end{definition}
The following concept is adopted from \cite{Gogic95thecomparative,DBLP:journals/ai/FrenchHIK13}.

\begin{definition}[Succinctness]\label{def:succint}
Let $\mathcal{C,C'}$ be two classes of  programs  and  for every problem $L$, $L\in \mathcal{C}\Leftrightarrow L\in \mathcal{C'}$.
 Say $\mathcal{C}$ is \emph{at least as succinct as} $\mathcal{C'}$ (i.e., $\mathcal{C'}\preceq \mathcal{C}$), if for every problem $L$,  $$\mbox{$L\in \mbox{Poly-}\mathcal{C'}$ $\Rightarrow$ $L\in \mbox{Poly-}\mathcal{C}$.}$$
 If  $\mbox{$L\in \mbox{Poly-}\mathcal{C}$ but $L\not\in \mbox{Poly-}\mathcal{C'}$}$ (i.e., $\mathcal{C}\not\preceq \mathcal{C'}$), then  $L$ \emph{separates} $\mathcal{C}$ from $\mathcal{C'}$.   If $\mathcal{C'}\preceq \mathcal{C}$ and $\mathcal{C}\not\preceq \mathcal{C'}$, then $\mathcal{C}$ is \emph{strictly more succinct than} $\mathcal{C'}$ (i.e., $\mathcal{C'}\prec \mathcal{C}$). Moreover, $\mathcal{C,C'}$ are \emph{succinctly incomparable}  if there is a problem $L$ separates $\mathcal{C}$ from $\mathcal{C'}$, and vice versa ( i.e.,$\mathcal{C}\not\preceq \mathcal{C'}$ and $\mathcal{C'}\not\preceq \mathcal{C}$).
\end{definition}
Please note that the above notions also apply to  formalisms like PF or boolean circuits, etc.
\subsection{The PARITY Problem}
The PARITY problem is defined as:
$$\mbox{PARITY}=\{\mbox{Binary strings with an odd number of}\ 1\mbox{'s}\}.$$
We may simply call a string in PARITY an odd  string, and PARITY$_n$  denotes the set of  odd  strings of length $n$. Observe that PARITY$_n$ contains $2^{n-1}$ strings.
It is not hard to see that PARITY$_n$ for $n=1,2$  can be computed by normal programs $\Pi_1=\{x_1\leftarrow \}$ and $\Pi_2=\{x_1\leftarrow not\ x_2, x_2\leftarrow not\ x_1\}$ respectively. Since $Ans(\Pi_1)=\{1\}$ (i.e., $\{x_1\}$), and $Ans(\Pi_2)=\{10, 01\}$ (i.e., $\{x_1\},\{x_2\}$). However, as stated below,  PARITY$_n$ for $n\geq 3$ have no representations in normal programs.
\begin{theorem}[PARITY$\notin$LP]\label{PARITY-NLP}
PARITY  cannot be represented by normal programs.
\end{theorem}
\begin{proof}
Suppose there is a normal program $\Pi_n$ that computes PARITY$_n$ for a fixed $n\geq 3$. Then  $\{x_1\}$ and $\{x_1,x_2,x_3\}$, which are two odd strings,  belong to $Ans(\Pi_n)$. However, this is impossible since it contradicts the \emph{anti-chain} property of $\Pi_n$ \cite{Lifschitz:2006:WSM:1131313.1131316}: if strings $I,I'\in Ans(\Pi_n)$ and $I\subseteq I'$ then $I=I'$.
\end{proof}

On the other hand, the anti-chain property is suppressed in CP. E.g.,  the answer set $111$ of program (\ref{SNE:SNE1}) is a superset of the other three
answer sets $100,010,001$. Clearly,  program (\ref{SNE:SNE1}) represents PARITY$_3$, moreover, it suggests a ``pattern'' for representing PARITY$_n$: The first part of the program
(e.g., the first two rules in (\ref{SNE:SNE1}))
generates all possible strings of $n-1$ bits, the second part identifies the last bit to produce an odd string.

Therefore, it is straightforward to give a sequence
of canonical programs $\{\Pi_n\}$ for PARITY$_n$. The following is a PARITY$_4$ program generated from the pattern:
\begin{equation}\label{equ:parity4-pattern}
\begin{array}{c}
x_1 \leftarrow   not\ not\ x_1,\\
x_2 \leftarrow   not\ not\ x_2,\\
x_3 \leftarrow   not\ not\ x_3,\\
\end{array}
 \begin{array}{c}
x_4  \leftarrow   x_1, x_2, not\ x_3,\\
x_4 \leftarrow   x_1, x_3, not\ x_2,\\
x_4 \leftarrow  x_2, x_3, not\ x_1,\\
x_4 \leftarrow  not\ x_1, not\ x_2, not\ x_3.\\
\end{array}
\end{equation}
Please note that the number of rules involved in the second part of the pattern grows exponentially, since the number of odd strings with the last bit $1$ grows exponentially.
\begin{theorem}[PARITY$\in$CP]\label{SNE:ub}
PARITY  can be represented by exponential size canonical programs.
\end{theorem}
By PF we denote propositional formulas built on classical connectives $\{\wedge,\vee,\neg\}$ with boolean variables. Related concepts like \emph{satisfaction}, \emph{model} etc., are defined as usual. By $M(\phi)$ we denoted the set of models of $\phi$. The \emph{size $|\phi|$ of a formula $\phi$} is the number of connectives occur in it. PARITY$_n$ for $n=1,2$ can be represented by formulas $x_1$ and $(x_1\wedge\neg x_2)\vee(\neg x_2\wedge x_1)$. Furthermore, it is a textbook result that PARITY$_n$ for $n\geq 3$ has polynomial size formulas in PF, i.e., PARITY$\in\mathsf{NC^1/poly}$ (or Poly-PF) \cite{ccama,DBLP:books/daglib/0028687}.

\section{Boolean Circuits, Completion and PARITY$_n$ Programs for $n\leq 2$}\label{sec:boolean-completion}
\subsection{Boolean Circuits}
A \emph{(boolean) circuit}  is a directed, cycle-free graph where each node is either a \emph{gate} marked  with one of $\{\wedge,\vee,\neg\}$ or a boolean variable.  The \emph{in-degree} (resp. \emph{out-degree}) of a node is called its \emph{fan-in} (resp. \emph{fan-out}). A node marked with a variable always has fan-in 0 and is called an \emph{input}. The \emph{output} of the  circuit is one gate designated with fan-out 0.

 The \emph{value} of a circuit $C_n$ under inputs $x_1,\ldots,x_n$, denoted by $C_n(x_1,\ldots,x_n)$, is the value of the output  obtained from  an iterative calculation  through the inputs and the intermediate gates in the usual way. The \emph{size $|C_n|$ of a circuit} $C_n$ is the number of gates occur in it. The \emph{depth} of a circuit is the length of the longest path from an input to the output. We say a circuit \emph{computes} (or \emph{represents}) a problem $L\subseteq\{0,1\}^n$, if
$w\in L\Leftrightarrow C_n(w)=1.$
\begin{figure}
\begin{center}
\includegraphics[scale=0.55]{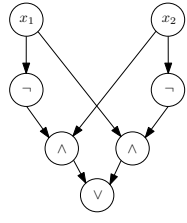}
\end{center}
\caption{A Parity$_2$ Circuit}\label{pic:parity2}
\end{figure}
E.g.,  a circuit $C_2$ that computes PARITY$_2$ is shown in Fig. \ref{pic:parity2}.
If  $L$ consists of strings of arbitrary lengths, then we introduce a sequence of circuits $\{C_n\} (n=1,2,\ldots$) to represent $L$, as indicated in Definition \ref{def:LPcomp}.

A circuit is said with \emph{bounded fan-in} if each gate has at most fain-in 2. If we do not have such restriction then the circuit is with \emph{unbounded fan-in}.
The \emph{class $\mathsf{AC^0}$} exactly contains all problems that can be computed  by a sequence of circuits $\{C_n\}$ in which the circuits  $C_n$ have  \emph{constant} depth and polynomial size $p(n)$.

 E.g., a sequence of polynomial size CNFs $\{\psi_n\}$ computes an $\mathsf{AC^0}$ language, in which a CNF is a \emph{conjunction} of \emph{clause} of the form
$(L_1\vee\ldots\vee L_m)$, where each $L_i$ is either a variable $x$ or a negated variable $\neg x$. Observe that  CNF has constant depth 2 ($\neg$ is usually not counted in the depth), and  each clause can be regarded as an unbounded fan-in gate $\vee$ with $m$ inputs. Note that $\{\psi_n\}$ cannot represent PARITY since PARITY$\notin\mathsf{AC^0}$.   For more details about circuits, please see \cite{ccama}.

\subsection{Completion and Related Theorems}
The \emph{completion} $Comp(\Pi)$ \cite{citeulike:411519,Erdem:2003:TLP:986809.986814} of  a canonical program $\Pi$,  consists of a set (or conjunction) of formulas\footnote{For convenience, we slightly abuse the connective $\equiv$ here.}:
\begin{itemize}
\item $x\equiv \tilde{B}_1\vee \tilde{B}_2\vee \cdots \vee \tilde{B}_m$, where $x\leftarrow B_1,\ldots,x\leftarrow B_m$ are all rules in $\Pi$ with head $x$, and each $\tilde{B}_i$ is
the conjunction of rule elements in $B_i$ with connective $not$ replaced by $\neg$,
\item $x\equiv\bot$, if $x$ is not a head of any rule in $\Pi$,
\item $\neg \tilde{B}$, if a rule $\bot\leftarrow B$ is in $\Pi$.
\end{itemize}
\begin{proposition}\label{prop:completion-poly-time}
Let $\Pi$ be an arbitrary canonical  program. Then $Comp(\Pi)$ is a   constant depth, unbounded fan-in circuit whose size  is polynomially bounded by $|\Pi|$.
\end{proposition}
\begin{proof}
All propositional formulas are circuits of fan-out 1, so $Comp(\Pi)$ is definitely  a circuit. Clearly,  its size  is  polynomially bounded by $|\Pi|$, and its depth  is a constant for arbitrary program $\Pi$. Moreover,
 there are no restrictions on the number of rule elements in a body or the number of rules in $\Pi$, therefore the corresponding gates  in $Comp(\Pi)$ are with unbounded fan-in.
\end{proof}
It is well-known that every answer set of a canonical program $\Pi$ is a \emph{model} of $Comp(\Pi)$,  but the inverse is generally not hold. E.g., the completion of the program $x\leftarrow x$ has two models $\{x\}$ and $\emptyset$, while it has a unique answer set $\emptyset$. It turns out that  $x\leftarrow x$ gives rise to a so-called \emph{loop}, which leads to an inappropriate model.
 It is shown in \cite{Lin2004115,loop-for-disjunctive}  that the so-called \emph{loop formulas}  $LF(\Pi)$ nicely eliminate inappropriate models of $Comp(\Pi)$, s.t. the models of  the union (or conjunction) of $LF(\Pi)$ and $Comp(\Pi)$ are coincided with $Ans(\Pi)$.

The \emph{(positive) dependency graph} \cite{Apt:1988:TTD:61352.61354}  of a canonical  program $\Pi$ is a pair $(N,E)$ in which the set of nodes  $N=var(\Pi)$, and $E$ contains a directed edge $(x,x')$ iff there is a rule  $H\leftarrow B$ in $\Pi$ s.t. $H=x$ and $x'\in B$. Note that rule elements of the form $not\ x'$ or $not\ not\ x'$ in $B$ do not contribute to the edges.
 A \emph{non-empty} set of variables $U\subseteq var(\Pi)$ is called a \emph{loop} of  $\Pi$, if i) $U$ is a \emph{singleton} $\{x\}$ and $(x,x)\in E$, or ii) $U$ is not a singleton and the restriction of the graph on $U$ is \emph{strongly connected}.

 Let $U$ be a loop of $\Pi$, define
$R^-(U,\Pi)=\{H\leftarrow B\in\Pi : H\in U,\ \neg\exists\mbox{ variable } x\in B\ \mbox{s.t.}\ x\in U \}.$
Let $\{B_1,\ldots,B_m\}$ be all the bodies of the rules in $R^-(U,\Pi)$,
then the \emph{loop formula} $LF(U,\Pi)$  is the following:
\begin{equation}\label{equ:LoopFormula}
\neg [\tilde{B}_{1}\vee\ldots\vee \tilde{B}_{m}]\supset \bigwedge_{x\in U}\neg x .
\end{equation}
$LF(\Pi)$ denotes the conjunction  of all loop formulas of $\Pi$.

\begin{theorem}[Lin-Zhao Theorem\cite{Lin2004115,loop-for-disjunctive}]\label{thm:lin-zhao-sne}
Let $\Pi$ be a canonical  program. Then   $\Pi$ is  equivalent to $Comp(\Pi)\cup LF(\Pi)$, i.e., $Ans(\Pi)=M(Comp(\Pi)\cup LF(\Pi))$.
\end{theorem}
By Theorem \ref{thm:lin-zhao-sne} (or the (generalized) Fages theorem \cite{fages94,Erdem:2003:TLP:986809.986814,You:2003:EAS:1630659.1630783}), if $\Pi$ has no loops, then $LF(\Pi)$ is a tautology $\top$ and  $\Pi$ is  equivalent to $Comp(\Pi)$ (i.e.,  \emph{completion-equivalent}).
\subsection{PARITY$_n$ Programs for $n\leq 2$}
\begin{proposition}\label{prop:cmp-loop-PARITY1}
Let $\Pi_1$ be a PARITY$_1$ canonical program. Then $\Pi_1$ is equivalent to $Comp(\Pi_1)$.
\end{proposition}
\begin{proof}
By Theorem \ref{thm:lin-zhao-sne}, the unique answer set $\{x_1\}$ of $\Pi_1$ is  a  model of $Comp(\Pi_1)\cup LF(\Pi_1)$, which also is a model of $LF(\Pi_1)$. There are two cases about the loops in $\Pi_1$: (i) $\Pi_1$ has no loops. $LF(\Pi_1)$ is simply $\top$; (ii) $\Pi_1$ has a singleton loop  $\{x_1\}$. Recall that $LF(\Pi_1)$ is a formula of the form $\neg [\tilde{B}_{1}\vee\ldots\vee \tilde{B}_{m}] \supset \neg x_1$, in which $B_{1} \ldots B_{m}$ are all the bodies of  rules in $R^-(\{x_1\},\Pi_1)$. In both cases, $\emptyset$ is a model of $LF(\Pi_1)$, so $LF(\Pi_1)$ is a  tautology. Therefore,  $\Pi_1$ is equivalent to $Comp(\Pi_1)$.
\end{proof}

Observe that Proposition \ref{prop:cmp-loop-PARITY1} does not hold for PARITY$_2$ programs. Consider the following PARITY$_2$ program:
\begin{equation}\label{equ:non-standard-parity2-1-1}
\begin{array}{c}
x_1 \leftarrow   not\  x_2,\\
x_2 \leftarrow   not\  x_1,\\
\end{array}
 \begin{array}{c}
x_1  \leftarrow   x_1,\\
x_2 \leftarrow   x_2.
\end{array}
\end{equation}
Clearly,  $\{x_1,x_2\}$ (i.e., $11$) is not an answer set of  (\ref{equ:non-standard-parity2-1-1}), but a model of its completion  $\{x_1\equiv x_1\vee \neg x_2,\  x_2\equiv x_2\vee \neg x_1\}$.

Note that  the rules $\{x_1 \leftarrow   x_1, x_2 \leftarrow   x_2 \}$ contribute to so-called \emph{singleton loops}. We may check that without the above two rules, program (\ref{equ:non-standard-parity2-1-1}) is  a completion-equivalent PARITY$_2$ program. In fact,  such  ``singleton loop'' rules can be always safely removed, as stated in  Proposition \ref{prop:equiv-non-singleton}.

Let $\Pi$ be a basic program and $I$ be a set of variables, define the Knaster-Tarski operator \cite{Apt:1988:TTD:61352.61354} as
$T_\Pi(I)=\{H\ :\ H\leftarrow B \in \Pi \ \mbox{and}\ I\models B\}$.
The operator  $T$ is monotone w.r.t. $I$ therefore has a  \emph{least fixed point} $T_\Pi^\infty(\emptyset)$, which can be computed by:
(i) $T_\Pi^0(\emptyset)  =  \emptyset$; (ii) $T_\Pi^{i+1}(\emptyset) =  T_\Pi(T_\Pi^i(\emptyset))$ and (iii) $T_\Pi^\infty(\emptyset)  =  \bigcup_{i\geq 0}(T_\Pi^i(\emptyset))$.
Moreover, $T$ is also monotone w.r.t. $\Pi$ for a given  $I$, i.e., $T_\Pi(I)\subseteq T_{\Pi'}(I)$ if $\Pi\subseteq \Pi'$. It is pointed out in \cite{Gelfond88thestable,You:2003:EAS:1630659.1630783} that $I\in Ans(\Pi)$ iff $I=T_{\Pi^I}^\infty(\emptyset)$ for   a canonical program $\Pi$.

\begin{proposition}\label{prop:equiv-non-singleton}
Let $\Pi$ be a canonical program. Then removing each rule $x\leftarrow B\in \Pi$ with $x\in var(B)$ results in a program $\Pi'$ s.t. $Ans(\Pi)=Ans(\Pi')$.
\end{proposition}
\begin{proof}
It is sufficient to show  $T_{\Pi^I}^\infty(\emptyset)=T_{\Pi'^I}^\infty(\emptyset)$ for any set $I$ of variables.
 Suppose $H\in T_{\Pi^I}^\infty(\emptyset)$ for some $I$, then  $\exists i>0$,  $H\in T_{\Pi^I}^i(\emptyset)$ and $H\notin T_{\Pi^I}^{i-1}(\emptyset)$. It is not hard to see that $H$ must be obtained from a rule  $H\leftarrow B$ in $\Pi$
 s.t. $H\notin var(B)$,  $H\leftarrow var(B)\in \Pi^I$ and  $T_{\Pi^I}^{i-1}(\emptyset)\models var(B)$. Note that $H\leftarrow B\in \Pi'$  and
 $H\leftarrow var(B)$ is in $\Pi'^I$ as well.
Now  we show  $H\in T_{\Pi'^I}^\infty(\emptyset)$. Suppose $H\in T_{\Pi^I}^1(\emptyset)$.
 So a rule $H\leftarrow $ is in $\Pi^I$ and $\Pi'^I$, clearly $H\in T_{\Pi'^I}^\infty(\emptyset)$.
Let $k>1$ and assume for all $i<k$, $T_{\Pi^I}^i(\emptyset)\subseteq T_{\Pi'^I}^\infty(\emptyset)$. Suppose $H'\in T_{\Pi^I}^k(\emptyset)$,
 then $\exists H'\leftarrow var(B')\in \Pi^I$ s.t. $T_{\Pi^I}^{k-1}(\emptyset)\models var(B')$.
  Obviously $H'\in T_{\Pi'^I}^\infty(\emptyset)$ since $H'\leftarrow var(B')\in \Pi'^I$ and $T_{\Pi'^I}^\infty(\emptyset)\models var(B')$ by induction hypothesis.
   Therefore $T_{\Pi^I}^\infty(\emptyset)\subseteq  T_{\Pi'^I}^\infty(\emptyset)$.
Note that  $\Pi'^I \subseteq \Pi^I$ since $\Pi'\subseteq \Pi$. It follows that $T_{\Pi'^I}^\infty(\emptyset)\subseteq T_{\Pi^I}^\infty(\emptyset)$ due to the monotonicity of operator $T$.
Hence $T_{\Pi^I}^\infty(\emptyset)=T_{\Pi'^I}^\infty(\emptyset)$.
\end{proof}

It turns out that we have a more general observation: deleting \emph{all rules} with variables in the body (thus removing all loops) does not affect the answer sets of a PARITY$_2$ program!

\begin{proposition}\label{prop:pseudo-loop-parity}
Let $\Pi_2$ be a PARITY$_2$ canonical program.  Then there is a PARITY$_2$ program $\Pi_2'$  which is equivalent to $Comp(\Pi_2')$ and  $|\Pi_2'|\leq |\Pi_2|$.
\end{proposition}
\begin{proof}
W.l.o.g., assume $\Pi_2$ has no singleton loops.  Let $\Pi_2'=\{H\leftarrow B\in\Pi_2:var(B)=\emptyset\}$,
 clearly $\Pi_2'\subseteq \Pi_2$ and thus $|\Pi_2'|\leq |\Pi_2|$. To see  $\Pi_2'$ is also a PARITY$_2$ program, it is sufficient to show for any $I\in Ans(\Pi_2)$, $Cn(\Pi_2^I)=Cn(\Pi_2'^I)$.
Suppose $H\in I$, i.e., $H\in Cn(\Pi_2^I)$. We claim that  $H$ must be obtained from a rule $H\leftarrow B$ in $\Pi_2$
s.t. (i) $I\models B$, and (ii)  $var(B)=\emptyset$. Clearly (i) holds. To see (ii), note that  $\Pi_2$ has exactly two answer sets $\{x_1\}$ and $\{x_2\}$. W.l.o.g., let $I=\{x_1\}$ thus  $H=x_1$.
Since $\Pi_2$ has no singleton loops, $x_1\notin var(B)$, and $x_2\notin var(B)$  since $I\models B$. Hence  $var(B)=\emptyset$.

Now it is easy to see $H\leftarrow B\in  \Pi_2'$ and $H\leftarrow \in \Pi_2'^I$ since $I\models B$ and $var(B)=\emptyset$.
Thus $H\in Cn(\Pi_2'^I)$. Therefore $Cn(\Pi_2^I)\subseteq Cn(\Pi_2'^I)$.
Since $\Pi_2'\subseteq \Pi_2$, we have  $Cn(\Pi_2'^I)\subseteq Cn(\Pi_2^I)$ due to the monotonicity of  operator $Cn(\cdot)$. Consequently, $Cn(\Pi_2^I)=Cn(\Pi_2'^I)$.
Observe that $\Pi_2'$ has no loops, so $\Pi_2'$ is equivalent to $Comp(\Pi_2')$ .
\end{proof}

Consider the following PARITY$_2$ program (\ref{equ:non-standard-parity2-5-1}), which has a non-singleton loop $\{x_1,x_2\}$ but not completion-equivalent. One may see that removing the two rules in the second line  makes it completion-equivalent, without affecting its answer sets.

\begin{equation}\label{equ:non-standard-parity2-5-1}
\begin{array}{l}
x_1 \leftarrow  not\ x_2,\\
x_1 \leftarrow   x_2, not\ not\ x_1,\\
\end{array}
 \begin{array}{l}
x_2 \leftarrow  not\ x_1,\\
x_2 \leftarrow   x_1, not\ not\ x_2.\\
\end{array}
\end{equation}

In the following, we shall introduce a general approach to simply an arbitrary PARITY program to be completion-equivalent.

\section{General Simplification of PARITY$_n$  Programs} \label{sec:gen-simply-parity}
Let $B$ be a set of rule elements  built on  associated variables  $V=\{x_1,\ldots,x_n\}$.
 We say $B$ is \emph{consistent} if there is a set of variables $I$ s.t. $I\models B$. Define $S(B)$ to be the set  $\{I\subseteq V :I\models B \}$. E.g., let $V=\{x_1,x_2,x_3,x_4\}$ and $B=\{x_2, not\ x_3, not\ not\ x_4\}$, then $B$ is consistent and $S(B)=\{\{x_1,x_2,x_4\},\{x_2,x_4\}\}=\{1101,0101\}$. Clearly,
if $B$ is not consistent then $S(B)=\emptyset$. Note that if a rule has an inconsistent body, then it is redundant and  can be safely removed.

We say $B$ \emph{covers} a variable $x\in V$ iff $x\in B$ or $not\ x\in B$ or $not\ not\ x\in B$. If $B$ covers every variable in $V$ then
$B$ \emph{fully covers} $V$. E.g.,
$B=\{x_1, not\ x_2, not\ not\ x_3\}$ fully covers $V=\{x_1,x_2,x_3\}$. Obviously,  $B$ is consistent and fully covers $V$ iff
 $S(B)$ contains a unique string.

 In the next section, we stipulate that the set of associated variables is $var(\Pi_n)$ whenever  $\Pi_n$ is the program under discussion,  we also  assume that a PARITY program has no singleton-loops and contains no inconsistent bodies.

\subsection{Simplifying Full Coverage Rules}
A rule $H\leftarrow B\in\Pi_n$ is a \emph{full coverage rule} if $B$ fully covers $var(\Pi_n)$.
\begin{lemma}\label{lema:parity-even-string-notnot}
Let $\Pi_n$ be a PARITY$_n$ program. Suppose there is a rule $x\leftarrow B$ in $\Pi_n$ s.t.
$not\ not\ x\in B$ and $S(B)$ contains a unique even string. Then removing $x\leftarrow B$ from $\Pi_n$  results in a PARITY$_n$ program $\Pi_n'$.
\end{lemma}
\begin{proof}
We show for any set $I$ of variables, $I=Cn(\Pi_n^I)$ iff $I=Cn(\Pi_n'^I)$.
Observe that  $\Pi_n'\subseteq \Pi_n$, then $Cn(\Pi_n'^I)\subseteq Cn(\Pi_n^I)$ for any $I$. So it is sufficient to show  $Cn(\Pi_n^I)\subseteq Cn(\Pi_n'^I)$.  Assume $Cn(\Pi_n^I)\nsubseteq Cn(\Pi_n'^I)$ for some $I$. It must be the case that $x\in Cn(\Pi_n^I)$ and
 $x\notin Cn(\Pi_n'^I)$ since $\Pi_n'\cup\{x\leftarrow B\}=\Pi_n$. Moreover, we have $x\leftarrow var(B)\in \Pi_n^I$ and $Cn(\Pi_n^I)\models var(B)$. The former implies that $I\models B\setminus var(B)$. Since $not\ not\ x\in B\setminus var(B)$, we have $I\models not\ not\ x$ (i.e., $x\in I$).
 
Now suppose $I=Cn(\Pi_n^I)$, then $I$ is an odd string. However, recall that $I\models B\setminus var(B)$ and $I=Cn(\Pi_n^I)\models var(B)$. Hence we have $I\models B$, i.e., $I\in S(B)$. This contradicts the fact that $I\in S(B)$ is an  even string. So  $Cn(\Pi_n^I)\subseteq Cn(\Pi_n'^I)$.

Suppose $I=Cn(\Pi_n'^I)$. As mentioned above,  $Cn(\Pi_n^I)\nsubseteq Cn(\Pi_n'^I)$ implies that  $x\notin Cn(\Pi_n'^I)$ and $x\in I$. However,  recall that $I=Cn(\Pi_n'^I)$, hence $x\in Cn(\Pi_n'^I)$, a contradiction.  So  $Cn(\Pi_n^I)\subseteq Cn(\Pi_n'^I)$.
\end{proof}

Note that Lemma \ref{lema:parity-even-string-notnot} also justifies our simplification for (\ref{equ:non-standard-parity2-5-1}).

\begin{lemma}\label{lema:parity-odd-string-notnot}
Let $\Pi_n$ be a PARITY$_n$ program. Suppose there is a rule $x\leftarrow B$ in $\Pi_n$ s.t.
$not\ not\ x\in B$ and $S(B)$ contains a unique odd string. Then replacing its body $B$  with  $B'=B\setminus\{not\ not\ x\}$
  results in a PARITY$_n$ program $\Pi_n'$.
\end{lemma}
\begin{proof}
We show that $I=Cn(\Pi_n^I)$ iff $I=Cn(\Pi_n'^I)$ for any set $I$ of variables.
 Suppose $I=Cn(\Pi_n^I)$, we shall prove $Cn(\Pi_n^I)=Cn(\Pi_n'^I)$. Consider the following cases:
 \begin{itemize}
 \item $I\models B\setminus var(B)$. Since $B'=B\setminus\{not\ not\ x\}$, clearly, $B'\setminus var(B')\subseteq B\setminus var(B)$, hence $I\models B'\setminus var(B')$. It follows that
  $x\leftarrow var(B)\in \Pi_n^I$ and $x\leftarrow var(B')\in \Pi_n'^I$.
 Furthermore, note that $\Pi_n\setminus\{x\leftarrow B\}=\Pi_n'\setminus\{x\leftarrow B'\}$ and $var(B)=var(B')$, thus  $\Pi_n^I=\Pi_n'^I$. So $Cn(\Pi_n^I)=Cn(\Pi_n'^I)$.
\item $I\nvDash B\setminus var(B)$. Consider the following subcases:
\begin{itemize}
\item $I\nvDash B'\setminus var(B')$. Similarly, we have $\Pi_n'^I=\Pi_n^I$, thus $Cn(\Pi_n'^I)=Cn(\Pi_n^I)$.
\item $I\models B'\setminus var(B')$. Clearly, in this case  $I\nvDash not\ not\ x$.
Now suppose $I\models var(B')$, so we have $I\models B'$. Recall that (i) $not\ not\ x\in B$, (ii) $x\notin B$ since $\Pi_n$ has no singleton loops, (iii) $B'=B\setminus\{not\ not\ x\}$
 and (iv) $S(B)$ contains a unique odd string, say $I'$.
It follows that  $S(B')=\{I', I'\setminus\{x\}\}$. Obviously $I$ must be $I'\setminus\{x\}$ since $I'\models not\ not\ x$. However, this is a contradiction since   $I'\setminus\{x\}$
 is an even string and  $I$  is an odd string since $I$ is an answer set of $\Pi_n$.
So suppose $I\nvDash var(B)$. Note that in this case $\Pi_n'^I=\Pi_n^I\cup\{x\leftarrow var(B)\}$, we show  $Cn(\Pi_n^I)=Cn(\Pi_n^I\cup\{x\leftarrow var(B)\})$, i.e.,  $Cn(\Pi_n^I)=Cn(\Pi_n'^I)$.
Firstly, $Cn(\Pi_n^I)\subseteq Cn(\Pi_n^I\cup\{x\leftarrow var(B)\})$ due to the monotonicity of operator
$Cn(\cdot)$. Assume $Cn(\Pi_n^I\cup\{x\leftarrow var(B)\})\nsubseteq Cn(\Pi_n^I)$, it must be  $Cn(\Pi_n^I)\models var(B)$ and $x\in Cn(\Pi_n^I\cup\{x\leftarrow var(B)\})$,
$x\notin Cn(\Pi_n^I)$. However this is impossible since $I=Cn(\Pi_n^I)$ and  $I\nvDash var(B)$. Therefore $Cn(\Pi_n^I\cup\{x\leftarrow var(B)\})\subseteq Cn(\Pi_n^I)$.
\end{itemize}

\end{itemize}

Suppose $I=Cn(\Pi_n'^I)$, we shall prove $Cn(\Pi_n^I)=Cn(\Pi_n'^I)$. Consider the following cases:
\begin{itemize}
\item $x\leftarrow var(B')\notin \Pi_n'^I$. Clearly, $I\nvDash B'\setminus var(B')$. So $I\nvDash B\setminus var(B)$ since $B'\setminus var(B')\subseteq B\setminus var(B)$. Therefore
$x\leftarrow var(B)\notin \Pi_n^I$ and we have $\Pi_n'^I=\Pi_n^I$. Hence $Cn(\Pi_n'^I)=Cn(\Pi_n^I)$.
\item $x\leftarrow var(B')\in \Pi_n'^I$. There are two subcases:
\begin{itemize}
\item $x\leftarrow var(B)\in \Pi_n^I$. Similarly, $\Pi_n'^I=\Pi_n^I$ and then $Cn(\Pi_n'^I)=Cn(\Pi_n^I)$.
\item $x\leftarrow var(B)\notin \Pi_n^I$. Clearly, $I\nvDash not\ not\ x$, i.e., $x\notin I$. Furthermore, recall that $\Pi_n'^I=\Pi_n^I\cup\{x\leftarrow var(B')\}$,
We shall show $Cn(\Pi_n'^I)=Cn(\Pi_n^I)$. Obviously $Cn(\Pi_n^I)\subseteq Cn(\Pi_n'^I)$.
Now assume $Cn(\Pi_n'^I)\nsubseteq Cn(\Pi_n^I)$. It must be the case that $x\in Cn(\Pi_n'^I)$ but $x\notin Cn(\Pi_n^I)$. However, since $I=Cn(\Pi_n'^I)$, we have $x\in I$, a contradiction.
So $Cn(\Pi_n'^I)\subseteq Cn(\Pi_n^I)$, therefore $Cn(\Pi_n'^I)=Cn(\Pi_n^I)$.
\end{itemize}
\end{itemize}
\end{proof}

\subsection{Standard PARITY$_n$ Programs}
A  PARITY$_n$ program $\Pi_n$ is \emph{standard} if for each rule $x\leftarrow B\in \Pi_n$, $not\ not\ x\notin B$  whenever $S(B\cup\{x\})$ contains a unique string.   E.g., the PARITY program (\ref{SNE:SNE1}) is standard, while (\ref{equ:non-standard-parity2-5-1}) is not. Note that if $\Pi_n$ is standard, then for any rule $x\leftarrow B \in\Pi_n$, $B$ does not cover $x$, i.e., $x\notin B$, $not\ x\notin B$ and $not\ not\ x\notin B$, since $\Pi_n$ has no singleton loops and $S(B\cup\{x\})$ is consistent.

\begin{proposition}\label{prop:standard-program}
Let $\Pi_n$ be a PARITY$_n$ program. Then there is a standard PARITY$_n$ program $\Pi'_n$ s.t. $|\Pi_n'|\leq |\Pi_n|$.
\end{proposition}
\begin{proof}
For each rule $x\leftarrow B\in \Pi_n$ in which $not\ not\ x\in B$ and $S(B)$ contains a unique string: (i) Delete $x\leftarrow B$ from $\Pi_n$ if $S(B)$ contains an even string; (ii) Remove $not\ not\ x$ from $B$ if $S(B)$ contains an odd string.
By  Lemma \ref{lema:parity-even-string-notnot} and \ref{lema:parity-odd-string-notnot}, the above procedure results in
a standard PARITY$_n$ program $\Pi_n'$ and  $|\Pi_n'|\leq |\Pi_n|$.
\end{proof}

\begin{proposition}\label{prop:standard-comp-equiv}
Let $\Pi_n$ be a standard PARITY$_n$ program. Then  $\Pi_n$ is equivalent to its completion $Comp(\Pi_n)$.
\end{proposition}
The proof idea of  Proposition \ref{prop:standard-comp-equiv} is that every standard PARITY$_n$ program $\Pi_n$ can be equivalently rewritten  to a loop-free program  $\Pi_n'$ by replacing each $x\in var(B)$ with $not\ not\ x$ for every  rule body $B$ in $\Pi_n$. By the Lin-Zhao Theorem  or the (generalized) \emph{Fages Theorem}, $\Pi_n'$ is equivalent to its completion $Comp(\Pi_n')$. And then the proposition follows from the fact that $Comp(\Pi_n')=Comp(\Pi_n)$, since $not$ is treated as classical negation $\neg$ in the completion.  The detailed proof is presented in subsection \ref{subsec-detailed-proof}.

\subsection{Proof of Proposition \ref{prop:standard-comp-equiv}}\label{subsec-detailed-proof}
For technical reasons, we divide the rewriting procedure into two steps, in the first step a standard PARITY program is converted to so-called \emph{almost pure} program and in the second step the program is converted to a  \emph{pure} one, i.e., a PARITY program that does not have any loops. Before doing so we show some lemmas.

\begin{lemma}\label{prop:fully-coverage-simplify-no-notnot}
Let $\Pi_n$ be a standard PARITY$_n$ program. For each rule $x\leftarrow B\in \Pi_n$, if $S(B\cup\{x\})$ contains a unique  string, then the string must be odd.
\end{lemma}
\begin{proof}
Since $\Pi_n$ is standard, $B$ does not cover $x$.  So we have $S(B)=\{I, I\setminus\{x\}\}$.  Assume $I$ is an even string,
then $I\setminus\{x\}$ must be an odd string. It follows that $I\setminus\{x\}$ is \emph{not} closed under $x\leftarrow B$, since $I\setminus\{x\}\models B$ but $I\setminus\{x\}\nvDash x$.
 However, $\Pi_n$ is a PARITY$_n$ program,  every odd string must be closed under $x\leftarrow B$. A contradiction.
\end{proof}

\begin{lemma}\label{lemma:body-coverage}
Let $\Pi_n$ be a PARITY$_n$ program.
\begin{enumerate}[(i)]
 \item If there is a  rule
  $x\leftarrow B\in \Pi_n$ s.t. $B$ is consistent and $B\cup\{x\}$ does not fully cover $var(\Pi_n)$, then $not\ not\ x\in B$.
\item If there is a  rule   $H\leftarrow B\in \Pi_n$ s.t. $B$ is consistent and $B\cup\{H\}$ is inconsistent, then $B$ fully covers $var(\Pi_n)$.
 \end{enumerate}
\end{lemma}
\begin{proof}
Note that for any rule $x\leftarrow B$ in a PARITY$_1$ program, $B\cup\{x\}$ must fully cover $var(\Pi_1)$ since $\Pi_1$ involves only one variable. So in the following we consider $n\geq 2$.

(i) Equivalently, we show that if $B$ is consistent and $not\ not\ x\notin B$,
 then  $B\cup\{x\}$ fully covers $var(\Pi_n)$. Assume  $B\cup\{x\}$ does not fully cover $var(\Pi_n)$.  It follows that $B$ covers $0\leq i<n$ variables in $var(\Pi_n)$ (i.e., $B$ does not fully cover $var(\Pi_n)$).
 Consider the following cases:
\begin{itemize}
    \item  $not\ x\in B$.  Note that $B$ is consistent and $n\geq 2$. It is not hard to see $S(B)$ has exactly $2^{n-i-1}\geq 1$ odd strings.
      It means there is at least one odd string  $I$, $I\models B$ and $I\nvDash x$. Therefore $I$ is not close under $x\leftarrow B$.
     However, since $\Pi_n$ is a PARITY$_n$ program, every odd
string must be closed under $x\leftarrow B$. A contradiction.
    \item $not\ x\notin B$.  $B$ does not cover $x$, since  $not\ not\ x\notin B$ and  $x\notin B$ for $\Pi_n$ has no singleton loops.
  Recall that $B$ is consistent and $n\geq 2$, thus $S(B)$ has exactly $2^{n-i-1}$  odd strings. Obviously, half of these strings do not satisfy $x$.
To be more precise, there are $2^{n-i-2}$ odd strings $I$, $I\models B$ and $I\nvDash x$. We have $2^{n-i-2}\geq 1$ since $i$ is at most $n-2$.
In other words, there is at least one odd string $I$ which is not close under $x\leftarrow B$. Again a contradiction.
\end{itemize}
Consequently, $B\cup\{x\}$ must fully cover $var(\Pi_n)$.

(ii) There are two cases about $H$:
    \begin{itemize}
    \item $H$ is $\bot$. Assume $B$ does not fully cover $var(\Pi_n)$, i.e.,   $B$ covers $i$   variables in $var(\Pi_n)$ with $0\leq i<n$. Since $B$ is consistent, it is easy to see
     $S(B)$ has exactly $2^{n-i-1}\geq 1$  odd  strings. So there exists at least one odd string $I$ is not closed under $\bot\leftarrow B$. A contradiction.
     \item $H$ is a variable $x\in var(\Pi_n)$. Since $B\cup\{x\}$ is inconsistent, we have $not\ x\in B$. It is not hard to see in this case $x\leftarrow B$ can be rewritten as $\bot\leftarrow B$. By an argument similar to the above, $B$ must fully cover  $var(\Pi_n)$.
    \end{itemize}
\end{proof}

\subsubsection{Almost Pure PARITY$_n$ Programs}
Let $\Pi_n$ be a standard PARITY$_n$ program in CP, by $F^-(\Pi_n)$ we denote  the set of rules $H\leftarrow B\in \Pi_n$ s.t. $B\cup\{H\}$ does not fully cover $var(\Pi_n)$, by
 $F^+(\Pi_n)$ we denote $\Pi_n\setminus F^-(\Pi_n)$. If for each  rule $H\leftarrow B\in F^+(\Pi_n)$ we have $var(B)=\emptyset$, then $\Pi_n$ is called \emph{almost pure}.

By Lemma \ref{lemma:body-coverage}, it is not hard to see that every rule  of the form $x\leftarrow B, not\ not\ x$ is  in $F^-(\Pi_n)$, and every rule of the form $\bot\leftarrow B$ or $x\leftarrow B, not\ x$ is  in $F^+(\Pi_n)$.

\begin{proposition}\label{prop:almost-pure-program}
Let $\Pi_n$ be a standard PARITY$_n$ program. Then there is an almost pure PARITY$_n$ program $\Pi_n'$ s.t. $|\Pi_n'|\leq |\Pi_n|$.
\end{proposition}
\begin{proof}
Let $B'$ be the set obtained from $B$ by replacing every $x\in var(B)$ with $not\ not\ x$. Note that $I\models B$ iff $I\models B'$ for any set of variables $I$.
Let $\Pi_n'$ be the program obtained from $\Pi_n$ by replacing every rule $H\leftarrow B\in F^+(\Pi_n)$  with $H\leftarrow B'$.
 Clearly  $\Pi_n'$ is almost pure and $|\Pi_n'|\leq |\Pi_n|$.
It remains to  prove that $\Pi_n'$ is also a PARITY$_n$ program, i.e., $I=T_{\Pi_n^I}^\infty(\emptyset)$ iff $I=T_{\Pi_n'^I}^\infty(\emptyset)$.

 ($\Rightarrow$) Suppose $I$ is an answer set of $\Pi_n$, i.e., $I=T_{\Pi_n^I}^\infty(\emptyset)$, we shall show $T_{\Pi_n^I}^\infty(\emptyset)=T_{\Pi_n'^I}^\infty(\emptyset)$:
 \begin{itemize}
 \item $T_{\Pi_n^I}^\infty(\emptyset)\subseteq T_{\Pi_n'^I}^\infty(\emptyset)$. Note that  $\bot\notin T_{\Pi_n^I}^\infty(\emptyset)$  since $T_{\Pi_n^I}^\infty(\emptyset)$ is an answer set. Suppose  $x\in T_{\Pi_n^I}^1(\emptyset)$, then $\exists x\leftarrow B\in \Pi_n$ s.t. $var(B)=\emptyset$ and $I\models B$. Clearly, we have $x\leftarrow B\in \Pi_n'$. It follows that   $x\leftarrow \in \Pi_n'^I$
  and then $x\in T_{\Pi_n'^I}^\infty(\emptyset)$. Let $k>1$ and assume  for all $i<k$, $T_{\Pi_n^I}^i(\emptyset)\subseteq T_{\Pi_n'^I}^\infty(\emptyset)$. Suppose
  $x\in T_{\Pi_n^I}^k(\emptyset)$ but $x\notin T_{\Pi_n^I}^{k-1}(\emptyset)$. Then $\exists x\leftarrow B\in\Pi_n$ s.t. $x\leftarrow var(B)\in\Pi_n^I$ and
  $T_{\Pi_n^I}^{k-1}(\emptyset)\models var(B)$. Observe that either $x\leftarrow B\in \Pi_n'$ or $x\leftarrow B'\in \Pi_n'$. The former implies that $x\leftarrow var(B)\in \Pi_n'^I$, clearly,
  $T_{\Pi_n'^I}^\infty(\emptyset)\models var(B)$ by induction hypothesis,  and thus $x\in T_{\Pi_n'^I}^\infty(\emptyset)$. The latter implies that $x\leftarrow \in \Pi_n'^I$,
  trivially, $x\in T_{\Pi_n'^I}^\infty(\emptyset)$. Therefore, $T_{\Pi_n^I}^\infty(\emptyset)\subseteq T_{\Pi_n'^I}^\infty(\emptyset)$.

 \item $T_{\Pi_n'^I}^\infty(\emptyset)\subseteq T_{\Pi_n^I}^\infty(\emptyset)$. We first show  $\bot\notin T_{\Pi_n'^I}^\infty(\emptyset)$. Assume
$\bot\in T_{\Pi_n'^I}^\infty(\emptyset)$, then $\exists \bot\leftarrow B_1\in \Pi_n'$ s.t. $I\models B_1\setminus var(B_1)$. Consider its source $\bot\leftarrow B$ in $\Pi_n$.
 Recall that $\Pi_n$ has no singleton loops and $B$ is consistent since $\Pi_n$ is standard. Furthermore, $B\cup\{\bot\}$ is inconsistent, then  $\bot\leftarrow B\in F^+(\Pi_n)$ by Lemma \ref{lemma:body-coverage} (ii). So $var(B_1)=\emptyset$, $I\models B_1$ and thus $I\models B$. The latter means that   $I$ is not closed under $\bot\leftarrow B\in\Pi_n$, which contradicts the fact that $I$ is an answer set of $\Pi_n$. So $\bot\notin T_{\Pi_n'^I}^\infty(\emptyset)$.
  Now suppose $x\in T_{\Pi_n'^I}^1(\emptyset)$, then  $\exists x\leftarrow B_1\in \Pi_n'$ s.t.
  $x\leftarrow \in \Pi_n'^I$ and  $I\models B_1$. Consider the source of $x\leftarrow B_1$:
 \begin{enumerate}[(i)]
\item  $x\leftarrow B_1\in \Pi_n$, $var(B_1)=\emptyset$.  Since $I\models B_1$, $x\leftarrow \in \Pi_n^I$, we have $x\in T_{\Pi_n^I}^\infty(\emptyset)$.
\item  $x\leftarrow B\in F^+(\Pi_n)$, $var(B)\neq \emptyset$ and $B_1=B'$. Note that $I\models B$ since $I\models B_1$. Furthermore, $I$
   is closed under $x\leftarrow B$ since $I$ is an answer set of $\Pi_n$. So $x\in I$, i.e, $x\in T_{\Pi_n^I}^\infty(\emptyset)$.
   \end{enumerate}
Suppose $x\in T_{\Pi_n'^I}^k(\emptyset)$ but $x\notin T_{\Pi_n'^I}^{k-1}(\emptyset)$ for some $k>1$.
It means that  $\exists x\leftarrow B_1\in \Pi_n'$ s.t. $var(B_1)\neq\emptyset$, $x\leftarrow var(B_1)\in \Pi_n'^I$,  $I\models B_1\setminus var(B_1)$  and
$T_{\Pi_n'^I}^{k-1}(\emptyset)\models var(B_1)$. Note that  $var(B_1)\neq\emptyset$ implies that $x\leftarrow B_1\in F^-(\Pi_n)$, $B_1\cup\{x\}$ does not fully cover $var(\Pi_n)$. By Lemma \ref{lemma:body-coverage} (i),
we have $not\ not\ x\in B_1$. Recall that $I\models B_1\setminus var(B_1)$, so $I\models not\ not\ x$, i.e., $x\in T_{\Pi_n^I}^\infty(\emptyset)$.
Therefore, $T_{\Pi_n'^I}^\infty(\emptyset)\subseteq T_{\Pi_n^I}^\infty(\emptyset)$.

\end{itemize}

($\Leftarrow$) Suppose $I$ is an answer set of $\Pi_n'$, i.e., $I=T_{\Pi_n'^I}^\infty(\emptyset)$, we shall show $T_{\Pi_n'^I}^\infty(\emptyset)=T_{\Pi_n^I}^\infty(\emptyset)$:
\begin{itemize}
\item $T_{\Pi_n'^I}^\infty(\emptyset)\subseteq T_{\Pi_n^I}^\infty(\emptyset)$. Note that $\bot\notin T_{\Pi_n'^I}^\infty(\emptyset)$. Suppose $x\in T_{\Pi_n'^I}^1(\emptyset)$,
then $\exists x\leftarrow B_1\in \Pi_n'$ s.t. $var(B_1)=\emptyset$, $x\leftarrow \in \Pi_n'^I$ and $I\models B_1$.
Now consider the   source of $x\leftarrow B_1$:
\begin{enumerate}[(i)]
\item   $x\leftarrow B_1\in \Pi_n$. So $x\leftarrow \in \Pi_n^I$ and clearly $x\in T_{\Pi_n^I}^\infty(\emptyset)$.
\item  $x\leftarrow B\in F^+(\Pi_n)$, $var(B)\neq\emptyset$ and $B_1=B$. Note that $I\models B_1\cup\{x\}$ since $x\in I$ and $I\models B_1$, it follows that $I\models B\cup\{x\}$.
Clearly, $B\cup\{x\}$ is consistent  and fully covers $var(\Pi_n)$. By Lemma \ref{prop:fully-coverage-simplify-no-notnot}, $I$ is exactly the unique odd string in $S(B\cup\{x\})$. Recall that $\Pi_n$ is a PARITY$_n$ program, so $I$ must be an answer set of $\Pi_n$, i.e., $I=T_{\Pi_n^I}^\infty(\emptyset)$. Therefore $x\in T_{\Pi_n^I}^\infty(\emptyset)$.
\end{enumerate}

Let $k>1$ and assume for all $i<k$, $T_{\Pi_n'^I}^i(\emptyset)\subseteq T_{\Pi_n^I}^\infty(\emptyset)$.
Suppose $x\in T_{\Pi_n'^I}^k(\emptyset)$ but $x\notin T_{\Pi_n'^I}^{k-1}(\emptyset)$. Then $\exists x\leftarrow B_1\in \Pi_n'$ s.t. $var(B_1)\neq\emptyset$, $x\leftarrow var(B_1)\in \Pi_n'^I$,  $I\models B_1\setminus var(B_1)$  and
$T_{\Pi_n'^I}^{k-1}(\emptyset)\models var(B_1)$. Note that $var(B_1)\neq\emptyset$ implies  $x\leftarrow B_1\in F^-(\Pi_n)$. Moreover, $x\leftarrow var(B_1)\in \Pi_n^I$. By inductive hypothesis,
$T_{\Pi_n^I}^\infty(\emptyset)\models var(B_1)$, therefore $x\in T_{\Pi_n^I}^\infty(\emptyset)$. Consequently, $T_{\Pi_n'^I}^\infty(\emptyset)\subseteq T_{\Pi_n^I}^\infty(\emptyset)$.

\item $T_{\Pi_n^I}^\infty(\emptyset)\subseteq T_{\Pi_n'^I}^\infty(\emptyset)$. We first show $x\in T_{\Pi_n^I}^\infty(\emptyset)$ implies $x\in T_{\Pi_n'^I}^\infty(\emptyset)$. Suppose  $x\in T_{\Pi_n^I}^1(\emptyset)$, then $\exists x\leftarrow B\in \Pi_n$, $var(B)=\emptyset$ and $I\models B$. It follows that $x\leftarrow B\in \Pi_n'$ and $x\leftarrow \in \Pi_n'^I$. Clearly, $x\in T_{\Pi_n'^I}^\infty(\emptyset)$.
Let $k>1$ and assume for all $i<k$, $x\in T_{\Pi_n^I}^i(\emptyset)$ implies $x\in T_{\Pi_n'^I}^\infty(\emptyset)$. Suppose $x\in T_{\Pi_n^I}^k(\emptyset)$ but $x\notin T_{\Pi_n^I}^{k-1}(\emptyset)$. Then $\exists x\leftarrow B\in \Pi_n$ s.t.  $x\leftarrow var(B)\in \Pi_n^I$, $var(B)\neq\emptyset$, $I\models B\setminus var(B)$ and $T_{\Pi_n^I}^{k-1}(\emptyset)\models var(B)$. By induction hypothesis,
$T_{\Pi_n'^I}^\infty(\emptyset)\models var(B)$, i.e., $I\models var(B)$.
Now  $I\models B$ since $I\models B\setminus var(B)$ and $I\models var(B)$, hence $I\models B'$. Observe that either $x\leftarrow B'\in \Pi_n'$ or $x\leftarrow B\in \Pi_n'$, in both cases
   $I\models x$ since $I$ is an answer set of $\Pi_n'$ and must be closed under every rule of $\Pi_n'$.Consequently, $x\in T_{\Pi_n^I}^\infty(\emptyset)$ implies $x\in T_{\Pi_n'^I}^\infty(\emptyset)$. It remains to show $\bot\notin T_{\Pi_n^I}^\infty(\emptyset)$. Assume $\bot\in T_{\Pi_n^I}^\infty(\emptyset)$, then $\exists \bot\leftarrow B$ in $F^+(\Pi_n)$ s.t.  $I\models B\setminus var(B)$ and $T_{\Pi_n^I}^k(\emptyset)\models var(B)$ for some $k\geq 1$. Notice that the latter means $I\models var(B)$, since variables $T_{\Pi_n^I}^k(\emptyset)\models var(B)$ implies $T_{\Pi_n'^I}^\infty(\emptyset)\models var(B)$ by the previous result and $I=T_{\Pi_n'^I}^\infty(\emptyset)$. So $I\models B$, i.e.,  $I\models B'$. However,   note that $\bot\leftarrow B'$ in $\Pi_n'$ and  $I$ is not closed under $\bot\leftarrow B'$. This contradicts the fact that  $I$ is an answer set of $\Pi_n'$. Therefore $\bot\notin T_{\Pi_n^I}^\infty(\emptyset)$, and hence $T_{\Pi_n^I}^\infty(\emptyset)\subseteq  T_{\Pi_n'^I}^\infty(\emptyset)$.
\end{itemize}
\end{proof}

\subsubsection{Pure PARITY$_n$ Programs}
Let $\Pi_n$ be an almost pure  PARITY$_n$ program. If for every  rule $H\leftarrow B\in \Pi_n$ we have $var(B)=\emptyset$, then  $\Pi_n$ is called \emph{pure}. Clearly, a pure program has no loops and is hence  completion-equivalent.

\begin{proposition}\label{prop:pure-program}
Let $\Pi_n$ be an almost pure PARITY$_n$ program. Then there is a pure PARITY$_n$ program $\Pi_n'$ s.t. $|\Pi_n'|\leq|\Pi_n|$.
\end{proposition}
\begin{proof}
A rule $H\leftarrow B$ is \emph{non-pure} if $var(B)\neq\emptyset$. We show by induction on the number $m$ of non-pure rules in $\Pi_n$.  Base step $m=0$,   the claim trivially holds. Let $m>0$ and assume the claim holds for all almost pure PARITY$_n$ programs containing $j<m$ non-pure rules. Suppose  $\Pi_n$ is an almost pure PARITY$_n$ program with $m$ non-pure rules, and let $H\leftarrow B$ be a non-pure rule in $\Pi_n$.  Note that  $H\leftarrow B$ must be in $F^-(\Pi_n)$ and $H$ is a variable $x$, since $\Pi_n$ is almost pure. Let $B'$ be obtained from $B$ by replacing each variable $x\in var(B)$ with $not\ not\ x$. Note that $I\models B$ iff $I\models B'$ for any set of variables $I$.
Let $\Pi$ be $\Pi_n\setminus \{H\leftarrow B\}$ and let  $\Pi_n''=\Pi\cup\{H\leftarrow B'\}$.
Clearly, $\Pi_n''$ is almost pure and $|\Pi_n''|\leq|\Pi_n|$. We shall show that $\Pi_n''$ is also a PARITY$_n$ program, i.e., $I=T_{\Pi_n^I}^\infty(\emptyset)$ iff $I=T_{\Pi_n''^I}^\infty(\emptyset)$.

 ($\Rightarrow$) Suppose $I=T_{\Pi_n^I}^\infty(\emptyset)$, we prove $T_{\Pi_n^I}^\infty(\emptyset)=T_{\Pi_n''^I}^\infty(\emptyset)$:
 \begin{itemize}
 \item $T_{\Pi_n^I}^\infty(\emptyset)\subseteq T_{\Pi_n''^I}^\infty(\emptyset)$. Note that  $\bot\notin T_{\Pi_n^I}^\infty(\emptyset)$  since $T_{\Pi_n^I}^\infty(\emptyset)$ is an answer set. Suppose $x\in T_{\Pi_n^I}^1(\emptyset)$,
 then $\exists x\leftarrow B_1\in \Pi_n$ s.t. $var(B_1)=\emptyset$ and $I\models B_1$. Clearly, $x\leftarrow B_1\in \Pi_n''$ and  $x\leftarrow \in \Pi_n''^I$,  thus $x\in T_{\Pi_n''^I}^\infty(\emptyset)$. Let $k>1$ and assume for all $i<k$, $T_{\Pi_n^I}^i(\emptyset)\subseteq T_{\Pi_n''^I}^\infty(\emptyset)$. Suppose
  $x\in T_{\Pi_n^I}^k(\emptyset)$ but $x\notin T_{\Pi_n^I}^{k-1}(\emptyset)$. Then $\exists x\leftarrow B_1\in\Pi_n$ s.t. $var(B_1)\neq\emptyset$, $x\leftarrow var(B_1)\in\Pi_n^I$ and $T_{\Pi_n^I}^{k-1}(\emptyset)\models var(B_1)$. Hence if $x=H, B_1=B$ then $x\leftarrow B'\in \Pi_n''$, otherwise $x\leftarrow B_1\in \Pi_n''$. The former implies  $x\leftarrow \in \Pi_n''^I$ since  $var(B')=\emptyset$, and $I\models B'$ due to $I\models B_1$, trivially, $x\in T_{\Pi_n''^I}^\infty(\emptyset)$. The latter implies  $x\leftarrow var(B_1)\in \Pi_n''^I$, clearly,
  $T_{\Pi_n''^I}^\infty(\emptyset)\models var(B_1)$ by induction hypothesis,  and thus $x\in T_{\Pi_n''^I}^\infty(\emptyset)$.  Therefore, $T_{\Pi_n^I}^\infty(\emptyset)\subseteq T_{\Pi_n''^I}^\infty(\emptyset)$.

 \item $T_{\Pi_n''^I}^\infty(\emptyset)\subseteq T_{\Pi_n^I}^\infty(\emptyset)$. We first show that $\bot\notin T_{\Pi_n''^I}^\infty(\emptyset)$. Assume
$\bot\in T_{\Pi_n''^I}^\infty(\emptyset)$. So $\exists \bot\leftarrow B_1\in \Pi_n''$ s.t. $\bot\leftarrow var(B_1)\in \Pi_n''^I$ and $I\models B_1\setminus var(B_1)$. Note that $\bot\leftarrow B_1$ must be in $F^+(\Pi_n)$ since $\Pi_n$ is almost pure. Hence $var(B_1)=\emptyset$, $I\models B_1$ thus  $I\models B$. The latter means that   $I$ is not closed under $\bot\leftarrow B_1\in\Pi_n$, which contradicts the fact that $I$ is an answer set of $\Pi_n$. So $\bot\notin T_{\Pi_n''^I}^\infty(\emptyset)$.  Suppose $x\in T_{\Pi_n''^I}^1(\emptyset)$, then  $\exists x\leftarrow B_1\in \Pi_n''$ s.t.   $x\leftarrow \in \Pi_n''^I$ and  $I\models B_1$. There are two cases about the source of $x\leftarrow B_1$:

(i)  $x\leftarrow B_1\in \Pi_n$, $var(B_1)=\emptyset$.  Since $I\models B_1$, so $x\leftarrow \in \Pi_n^I$,  $x\in T_{\Pi_n^I}^\infty(\emptyset)$.

  (ii)$x\leftarrow B_1$ is obtained from $H\leftarrow B\in F^-(\Pi_n)$, i.e., $H=x$, $var(B)\neq \emptyset$ and $B_1=B'$. Note that $I\models B$ since $I\models B_1$. Recall that $not\ not\ x\in B$, so $x\in I$, i.e, $x\in T_{\Pi_n^I}^\infty(\emptyset)$.

Now suppose $x\in T_{\Pi_n''^I}^k(\emptyset)$ but $x\notin T_{\Pi_n''^I}^{k-1}(\emptyset)$ for some $k>1$.
It means that $\exists x\leftarrow B_1\in \Pi_n''$ s.t. $var(B_1)\neq\emptyset$, $x\leftarrow var(B_1)\in \Pi_n''^I$,  $I\models B_1\setminus var(B_1)$  and
$T_{\Pi_n''^I}^{k-1}(\emptyset)\models var(B_1)$. Since $\Pi_n$ is almost pure, $var(B_1)\neq\emptyset$ implies that $x\leftarrow B_1\in F^-(\Pi_n)$. By Lemma \ref{lemma:body-coverage} (i),  $not\ not\ x\in B_1$. Recall that $I\models B_1\setminus var(B_1)$, so $I\models not\ not\ x$, i.e., $x\in T_{\Pi_n^I}^\infty(\emptyset)$.
Therefore, $T_{\Pi_n''^I}^\infty(\emptyset)\subseteq T_{\Pi_n^I}^\infty(\emptyset)$.
\end{itemize}

($\Leftarrow$) Suppose $I=T_{\Pi_n''^I}^\infty(\emptyset)$, we  show  $T_{\Pi_n''^I}^\infty(\emptyset)=T_{\Pi_n^I}^\infty(\emptyset)$.
\begin{itemize}
\item  $T_{\Pi_n''^I}^\infty(\emptyset)\subseteq T_{\Pi_n^I}^\infty(\emptyset)$. Note that $\bot\notin T_{\Pi_n''^I}^\infty(\emptyset)$. Suppose  $x\in T_{\Pi_n''^I}^1(\emptyset)$,  then $\exists x\leftarrow B_1\in \Pi_n''$ s.t. $x\leftarrow \in \Pi_n''^I$ and $I\models B_1$.
Now consider the  source of $x\leftarrow B_1$:

(i)  $x\leftarrow B_1\in \Pi_n$ and $var(B_1)=\emptyset$. It follows that $x\leftarrow \in \Pi_n^I$ and clearly $x\in T_{\Pi_n^I}^\infty(\emptyset)$.

 (ii) $x\leftarrow B_1$ is obtained from $H\leftarrow B\in F^-(\Pi_n)$, i.e., $H=x$ and $B_1=B'$. So $x\leftarrow var(B) \in\Pi_n^I$ and $I\models B$. The latter implies $T_{\Pi_n''^I}^\infty(\emptyset)\models B$ thus $T_{\Pi_n''^I}^\infty(\emptyset)\models var(B)$. Note that $\Pi_n''^I=\Pi^I\cup\{x\leftarrow \}$, and  $x\notin var(B)$ since $\Pi_n$ is almost pure, it has no singleton loops. We prove by induction that  $T_{\Pi^I}^\infty(\emptyset)\models var(B)$. Suppose $x'\in var(B)$ and $x'\in T_{\Pi_n''^I}^1(\emptyset)$, then $\exists x'\leftarrow B_2\in \Pi''$, $var(B_2)=\emptyset$ and $I\models B_2$. Clearly, $x'\leftarrow \in \Pi^I$, thus $x'\in T_{\Pi^I}^\infty(\emptyset)$. Let $s>1$ and assume for all $t<s$, $x'\in var(B)$ and $x'\in T_{\Pi_n''^I}^t(\emptyset)$ implies $x'\in T_{\Pi_n^I}^\infty(\emptyset)$. Suppose $x'\in var(B)$ and $x'\in T_{\Pi_n''^I}^s(\emptyset)$. Then $\exists x'\leftarrow B_2\in \Pi''$,  $x'\leftarrow var(B_2)\in\Pi''^I$ and $T_{\Pi_n''^I}^{s-1}(\emptyset)\models var(B_2)$. Recall that $x\neq x'$ since $x\notin var(B)$, so  $x'\leftarrow var(B_2)\in \Pi^I$. By induction hypothesis, $T_{\Pi_n^I}^\infty(\emptyset)\models var(B_2)$,  $x'\in T_{\Pi_n^I}^\infty(\emptyset)$.
 Hence  $T_{\Pi^I}^\infty(\emptyset)\models var(B)$. Furthermore, $\Pi_n^I=\Pi^I\cup\{x\leftarrow var(B)\}$, so $T_{\Pi^I}^\infty(\emptyset)\subseteq T_{\Pi_n^I}^\infty(\emptyset)$,
    $T_{\Pi_n^I}^\infty(\emptyset)\models var(B)$. Therefore $T_{\Pi_n^I}^\infty(\emptyset)\models x$, i.e., $x\in T_{\Pi_n^I}^\infty(\emptyset)$.

Let $k>1$ and assume for all $i<k$, $T_{\Pi_n''^I}^i(\emptyset)\subseteq T_{\Pi_n^I}^\infty(\emptyset)$. Suppose $x\in T_{\Pi_n''^I}^k(\emptyset)$, then   $\exists x\leftarrow B_1\in \Pi_n''$, $var(B_1)\neq\emptyset$,  $x\leftarrow var(B_1)\in\Pi_n''^I$, $T_{\Pi_n''^I}^{k-1}(\emptyset)\models var(B_1)$. Clearly,  $x\leftarrow var(B_1)\in \Pi^I$. It follows that $x\leftarrow var(B_1)\in \Pi_n^I$ since $\Pi^I\subseteq \Pi_n^I$. By induction hypothesis,
$T_{\Pi_n^I}^\infty(\emptyset)\models var(B_1)$. Hence $T_{\Pi_n^I}^\infty(\emptyset)\models x$, i.e., $x\in T_{\Pi_n^I}^\infty(\emptyset)$. Therefore, $T_{\Pi_n''^I}^\infty(\emptyset)\subseteq T_{\Pi_n^I}^\infty(\emptyset)$.

\item  $T_{\Pi_n^I}^\infty(\emptyset)\subseteq T_{\Pi_n''^I}^\infty(\emptyset)$. We first show for any variable $x$,  $x\in T_{\Pi_n^I}^\infty(\emptyset)$ implies $x\in T_{\Pi_n''^I}^\infty(\emptyset)$. Suppose $x\in T_{\Pi_n^I}^1(\emptyset)$, then
$\exists x\leftarrow B_1\in \Pi_n$, $var(B_1)=\emptyset$ and $I\models B_1$. It follows that $x\leftarrow B_1\in \Pi_n''$ and $x\leftarrow \in \Pi_n''^I$. So $x\in T_{\Pi_n''^I}^\infty(\emptyset)$.
Let $k>1$ and assume for all $i<k$, $x\in T_{\Pi_n^I}^i(\emptyset)$ implies $x\in T_{\Pi_n''^I}^\infty(\emptyset)$. Suppose $x\in T_{\Pi_n^I}^k(\emptyset)$, then $\exists x\leftarrow B_1\in \Pi_n$ s.t.  $x\leftarrow var(B_1)\in \Pi_n^I$, $var(B_1)\neq\emptyset$, $I\models B_1\setminus var(B_1)$ and $T_{\Pi_n^I}^{k-1}(\emptyset)\models var(B_1)$. By induction hypothesis,
$T_{\Pi_n''^I}^\infty(\emptyset)\models var(B_1)$, i.e., $I\models var(B_1)$.
Now $I\models B_1$ since $I\models B_1\setminus var(B_1)$ and $I\models var(B_1)$, hence $I\models B_1$. Observe that if $H=x$ and $B_1=B$ then $x\leftarrow B'\in\Pi_n''$, otherwise $x\leftarrow B_1\in\Pi_n''$. In both cases  $I\models x$ since $I$ is an answer set of $\Pi_n''$ and it must be closed under every rule of $\Pi_n''$.
Consequently, $x\in T_{\Pi_n^I}^\infty(\emptyset)$ implies $x\in T_{\Pi_n''^I}^\infty(\emptyset)$. It remains to show $\bot\notin T_{\Pi_n^I}^\infty(\emptyset)$.
Assume $\bot\in T_{\Pi_n^I}^\infty(\emptyset)$, then  $\exists \bot\leftarrow B_1$ in $F^+(\Pi_n)$ s.t. $\bot\leftarrow \in\Pi_n^I$ and
 $I\models B_1$. Notice that $var(B_1)=\emptyset$ since $\Pi_n$ is almost pure. Furthermore, $\bot\leftarrow B_1$ must be in $\Pi_n''$, therefore $I\models \bot$ since  $I\models B_1$.However, this contradicts the fact that  $I$ is an answer set of $\Pi_n''$. Therefore $\bot\notin T_{\Pi_n^I}^\infty(\emptyset)$, and hence $T_{\Pi_n^I}^\infty(\emptyset)\subseteq  T_{\Pi_n''^I}^\infty(\emptyset)$.
\end{itemize}
Consequently, $\Pi_n''$ is an almost pure PARITY$_n$ program with $m-1$ non-pure rules. By induction hypothesis, there is a pure  PARITY$_n$ program $\Pi_n'$ with $|\Pi_n'|\leq|\Pi_n''|\leq |\Pi_n|$.
\end{proof}

\subsection{The Main Results}
The  main lemma below easily follows from  Proposition \ref{prop:standard-program} and \ref{prop:standard-comp-equiv}.
\begin{lemma}[Main Lemma]\label{lma:main-lemma}
Let $\Pi_n$ be a PARITY$_n$ program. Then there is a PARITY$_n$ program $\Pi_n'$ s.t. $\Pi_n'$ is equivalent to $Comp(\Pi_n')$ and  $|\Pi_n'|\leq|\Pi_n|$.
\end{lemma}
\begin{theorem}[PARITY$\notin$Poly-CP]\label{thm:main-theorem}
PARITY has no polynomial size CP programs.
\end{theorem}
\begin{proof}
Assume the contrary that there is a sequence of programs $\{\Pi_n\}$ in CP which represents PARITY, and $|\Pi_n|$ is bounded by a polynomial $p(n)$. By Lemma \ref{lma:main-lemma}, there is a sequence of completion-equivalent PARITY programs $\{\Pi_n'\}$ in which  $|\Pi_n'|$ is also bounded by the polynomial $p(n)$. By Proposition \ref{prop:completion-poly-time}, $\{\Pi_n'\}$ represents a language in $\mathsf{AC^0}$. This contradicts PARITY$\notin\mathsf{AC^0}$.
\end{proof}
\begin{corollary}
PARITY separates PF from CP.
\end{corollary}
\begin{corollary}\label{main-corollary}
Suppose $\mathsf{P}\nsubseteq\mathsf{NC^1/poly}$. Then CP and PF are succinctly incomparable.
\end{corollary}

\section{Discussion and Some More Results}\label{sec:more-suc-result}
 Interestingly, our main result may at first appear counter-intuitive: the \textsf{P}-complete problem PATH has Poly-CP representations, while this does not hold for an ``easy''  problem PARITY. Actually, there is no contradiction. As noted in \cite{DBLP:books/aw/AbiteboulHV95,Dantsin:2001:CEP:502807.502810}, a \emph{complete} problem in a complexity class  can be represented in a formalism $\mathcal{C}$, \emph{does not} imply that \emph{all} problems in that class can be represented in $\mathcal{C}$.

Generally speaking, the research of succinctness \cite{Gogic95thecomparative,DBLP:conf/kr/Coste-MarquisLLM04,DBLP:journals/lmcs/GroheS05,DBLP:journals/ai/FrenchHIK13} gives us a deeper understanding about KR formalisms, for it reveals their (in)abilities of concisely representing different problems under the condition that the encoded models are the \emph{same}. In terms of the theory of computation, succinctness essentially concerns with the computational power of different formalisms (i.e., models of computation).  This is particularly interesting if the formalisms are \emph{equally expressive} and share the same reasoning complexity. E.g., logic programs with \emph{cardinality constraints} and \emph{choice rules} (CC, without classic negation $\neg$) \cite{DBLP:journals/ai/SimonsNS02}, \emph{(simple) definite causal theories} (S/DT) \cite{Giunchiglia04nonmonotoniccausal} and \emph{two-valued} programs (TV) \cite{DBLP:conf/iclp/Lifschitz12} are as expressive as PF and \textsf{NP}-complete for consistency checking. But they have a non-trivial succinctness picture, see Fig. \ref{pic:succ}.

\begin{figure}
\begin{center}
\includegraphics[scale=0.7]{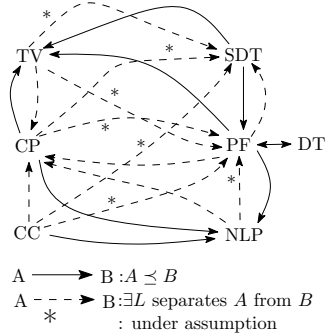}
\end{center}
\caption{A Succinctness Picture}\label{pic:succ}
\end{figure}

Besides the theoretical interests, succinctness also tells us something like ``which for what is the best'' in choosing KR formalisms for a given application. E.g.,  one should choose ASP instead of SAT  if the application involves reasoning about PATH or Transitive Closure\footnote{An $\mathsf{NL}$-\emph{complete} problem. It is believed that $\mathsf{NL\nsubseteq NC^1/Poly}$.}, because the former provides compact representations to avoid unnecessary overload. Recall that from the complexity viewpoint, even \emph{one} extra variable may \emph{double} the search space for intractable problems.

In the following we shall briefly discuss some succinctness results illustrated in Fig. \ref{pic:succ}, note that all mentioned formalisms have the same expressive power and same reasoning complexity.
\subsection{Logic Programs with Cardinality Constraints (CC)}
Simply speaking, CC extends normal programs (LP) with so-called  \emph{cardinality constraints} and \emph{choice rules} \cite{DBLP:journals/ai/SimonsNS02}. A \emph{choice rule}
\begin{equation}\label{equ:simple-choice}
\{x\}\leftarrow
\end{equation}
has two answer sets $\{x\}$ and $\emptyset$, i.e., same as  $x\leftarrow not\ not\ x$. Moreover, a choice
rule $\{x_1,\ldots,x_n\}\leftarrow$ produces $2^n$ answer sets, i.e., all subsets of $\{x_1,\ldots,x_n\}$. A \emph{cardinality constraint} is an expression of the form
\begin{equation}\label{equ:cardinality-constraint}
l\leq B \leq u
\end{equation}
in which $B$ is a finite set of rule elements of the form $x$ or $not\ x$, and integer $l$ (resp. $u$) is the lower (resp. upper) bound on $B$. In this paper we assume the magnitude of $l$ (and $u$) is polynomially bounded by $n$.

Intuitively, a set of variables $I$ satisfies (\ref{equ:cardinality-constraint}), if the number of satisfied rule elements in $B$ fulfills the related bounds. E.g., $\{x_1\}$ satisfies $1\leq\{x_1,x_2,x_3\}\leq 1$ but not $2\leq\{x_1,x_2,x_3\}\leq 3$, while $\{x_2,x_3\}$ satisfies the latter. Informally, we may think of (\ref{equ:cardinality-constraint}) as a special kind of rule element, and the answer set semantics is defined accordingly.

The following is a PARITY$_3$ program in CC:
\begin{equation}\label{equ:parity3-cc}
\begin{array}{l}
\{x_1,x_2,x_3\} \leftarrow  ,\\
\bot \leftarrow  0\leq\{x_1,x_2,x_3\}\leq 0,\\
\bot \leftarrow  2\leq\{x_1,x_2,x_3\}\leq 2,\\
\end{array}
\end{equation}
Clearly, the pattern applies to all PARITY$_n$ and the program  grows linearly. We define the size of a CC program to be the number of cardinality constraints occur in it.
\begin{theorem}[PARITY$\in$Poly-CC]\label{thm:cc-parity}
PARITY has polynomial size programs in CC.
\end{theorem}

An equivalent translation from CC to NLP was presented in \cite{Ferraris:2005:WCN:1041209.1041211}, however, the translation may involve exponential size blowup, since every cardinality constraint is simply converted to a formula  via a brute force enumeration. In fact, such a translation can be reduced to be  polynomial by adopting a non-trivial, sophisticated encoding for so-called \emph{threshold functions}\footnote{E.g., see Chapter 2 of \cite{savage1998models}.}. Therefore, we have:

\begin{theorem}\label{thm:cc-nlp}
NLP is at least as succinct as CC.
\end{theorem}

\subsection{Definite Causal Theories (DT)}
A variable $x$ or negated variable $\neg x$ is called a \emph{literal}.
A \emph{definite (causal) theory}  $D_n$ on \emph{signature} $\{x_1,\ldots,x_n\}$ is a finite set of (causal) rules of the form
\begin{equation}
H\Leftarrow G
\end{equation}
in which $H$ is either a literal or $\bot$, and $G$ is a propositional formula. If every $G$ is a conjunction of variables or negated variables, then $D_n$ is called \emph{simple} (SDT)\footnote{SDT is originally named as \emph{Objective Programs} in \cite{mccain97thesis}.}.

 The \emph{reduct} $D_n^I$ of $D_n$ w.r.t. a set of variables $I$, is the set of the heads $H$ of all rules in $D_n$ whose bodies $G$ are satisfied by $I$. Say $I$ is a \emph{model} of $D_n$ if $I$ is the unique model of $D_n^I$. The following theory:
\begin{equation}\label{equ:dt-choice}
\begin{array}{cc}
x\Leftarrow x, & \neg x\Leftarrow \neg x
\end{array}
\end{equation}
has two models $\{x\}$ and $\emptyset$, which is equivalent to program $x\leftarrow not\ not\ x$ or $\{x\}\leftarrow$.

If a definite theory $D_n$ is simple, then its size $|D_n|$  is defined as  the number of rules in it, otherwise $|D_n|$ is the number of connectives in it.
It is well-known that $D_n$ is equivalent to its (\emph{literal}) completion $Comp(D_n)$, in which $Comp(D_n)$ is similarly defined as for logic programs \cite{mccain97thesis,Giunchiglia04nonmonotoniccausal}.
It means that definite theories are fragments of PF, i.e., DT$\preceq$ PF. Therefore, the problems that can be represented by Poly-DT are in $\mathsf{NC^1/poly}$ as well. Moreover,  the completion of a simple definite theory  is also a constant depth, unbounded fan-in circuit whose size is polynomially bounded. By a proof similar to that of Theorem \ref{thm:main-theorem}, we have the following theorem:
\begin{theorem}[PARITY$\notin$Poly-SDT]
PARITY has no polynomial size theories in SDT.
\end{theorem}

Consider the (non-simple) causal theory (\ref{dc:parity3}) for  PARITY$_2$, where the body of the last rule is the negation of a PARITY$_2$ formula:
\begin{equation}\label{dc:parity3}
\begin{array}{l}
x_1\Leftarrow x_1,\ \ \   \neg x_1\Leftarrow \neg x_1, \\
 x_2\Leftarrow x_2,\ \ \  \neg x_2\Leftarrow \neg x_2, \\
\bot\Leftarrow  \neg((x_1\wedge\neg x_2)\vee(\neg x_2\wedge x_1)).
\end{array}
\end{equation}
Recall that PARITY have polynomial formulas in PF, therefore it is not hard to see  we can have polynomial DT theory for PARITY by the above pattern.

\begin{theorem}[PARITY$\in$Poly-DT]\label{thm:parity-polygdt}
PARITY has  polynomial size theories in DT.
\end{theorem}

Since PATH is $\mathsf{P}$-complete \cite{Lifschitz:2006:WSM:1131313.1131316}, therefore if PATH has polynomial representations in Poly-DT, then $\mathsf{P}\subseteq\mathsf{NC^1/poly}$, which is  believed impossible.

\begin{theorem}[PATH$\notin$Poly-DT] \label{them:path-not-dt} Suppose $\mathsf{P}\nsubseteq\mathsf{NC^1/poly}$.
Then PATH has  no polynomial size definite theories.
\end{theorem}

By the fact that PATH has polynomial size CP programs, we have:
\begin{corollary} \label{them:LP-DT-In}
Suppose $\mathsf{P}\nsubseteq\mathsf{NC^1/poly}$. Then CP and DT are succinctly incomparable.
\end{corollary}

It is worth  to point out that  some  difficulties observed in the literature could be nicely  explained by the above succinctness results. E.g., DT has been observed hard to concisely encode Transitive Closure (TC) \cite{Giunchiglia04nonmonotoniccausal,almostdef}. Recall that  Poly-DT represents problems in $\mathsf{NC^1/Poly}$, and  TC is a problem in $\mathsf{NC^2/poly}$ \cite{p-complete-book}, a class widely believed strictly contains $\mathsf{NC^1/poly}$. So unless the two classes coincide, TC has no polynomial size definite theories.

\subsection{Two-Valued Logic Programs (TV)}
A (\emph{two-valued}) program \cite{DBLP:conf/iclp/Lifschitz12} $\Pi_n$ on signature $\{x_1,\ldots,x_n\}$ is a finite set of (two-valued) rules of the form:
\begin{equation}\label{equ:tv-rule}
H\leftarrow B: G
\end{equation}
in which  $B\cup\{H\}$ is a finite set of literals and  $G$ is a formula. The \emph{reduct} $\Pi_n^I$ of $\Pi_n$ w.r.t. a set of variables $I$, is the set of rules
\begin{equation}\label{equ:tv-rule-reduct}
H\leftarrow B
\end{equation}
from $\Pi_n$ s.t. $I$ satisfies $G$. A set of literals $J$ is \emph{closed} under rule (\ref{equ:tv-rule-reduct}) if $H\in J$ whenever $B\subseteq J$. We say $I$ is a \emph{model} of $\Pi_n$ if $I$ is the \emph{unique} model of the minimal closure $J$ under every rule of $\Pi_n^I$. The following program $\Pi_2$ in TV
\begin{equation}\label{tv:choice}
\begin{array}{c}
x\leftarrow :  x,\ \ \  \ \ \  \neg x\leftarrow : \neg x
\end{array}
\end{equation}
has two models $\{x\}$ and $\emptyset$, which is equivalent to (\ref{equ:dt-choice}).

The following observations were pointed out in \cite{DBLP:conf/iclp/Lifschitz12}.
A formula $\phi_n$ can be rewritten in TV ($i\in\{1,\ldots,n\}$)\footnote{$\bot\leftarrow: \neg \phi_n$ is a shorthand of $x_1\leftarrow: \neg \phi_n,\ \  \neg x_1\leftarrow: \neg \phi_n$.}:
\begin{equation}\label{tv:2pf}
\begin{array}{c}
x_i\leftarrow :  x_i,\ \ \ \neg x_i\leftarrow : \neg x_i, \ \  \ \bot\leftarrow: \neg \phi_n.
\end{array}
\end{equation}
A causal rule $H\Leftarrow G$ can be equivalently rewritten as $H\leftarrow :G$. Moreover, to equivalently rewrite a CP program $\Pi_n$, each rule:
\begin{equation}\label{tv:cp2tv}
\begin{array}{cl}
H\leftarrow u_{1},\ldots, u_{j}, & not\ y_{j+1},\ldots, not\ y_{m},  \\
& not\ not\ z_{m+1},\ldots, not\ not\ z_{k}
\end{array}
\end{equation}
can be  translated  as:
\begin{equation}\label{tv:cp2tv-1}
\begin{array}{c}
H\leftarrow u_{1},\ldots, u_{j}: \neg y_{j+1}\wedge, \ldots \neg y_{m}\wedge z_{m+1}\wedge \ldots z_{k}  \\
\end{array}
\end{equation}
and add $\neg x\leftarrow :\neg x$ for every $x\in var(\Pi_n)$. All together, we have:
\begin{theorem}\label{thm:tv}
Two valued programs are strictly more succinct than: (i) propositional formulas and definite  theories, if $\mathsf{P}\nsubseteq\mathsf{NC^1/poly}$; (ii) canonical programs.
\end{theorem}

\section{Conclusions}
The main result of the paper is that the PARITY problem separates PF from CP, i.e., PARITY has no polynomial size CP programs, but  has  polynomial size PF formulas. Together with Lifschitz and Razborov's separation result, i.e., there exists a problem separates CP from PF (assuming $\mathsf{P}\nsubseteq \mathsf{NC^1/poly}$), we conclude that the two well-known KR formalisms are succinctly incomparable. In other words, if we consider CP and PF as two different models of computation, the above result just states that they are incomparable in terms of computational power. We also give a non-trivial succinctness picture on a family of logic program classes which posses the same expressive power and same reasoning complexity as PF.

 In  future work, we plan to investigate some missing connections in Fig. \ref{pic:succ}, e.g., we conjecture that there is a problem separates NLP from CP, SDT and CP are succinctly incomparable.

\subsection*{Acknowledgement}
We are grateful to the anonymous reviewers for their valuable comments. Thanks to Shiguang Feng, Yan Zhang, Jiankun He,  Guangrui Dang and Xiaolong Liang for their helpful discussions. The research was partially supported by NSFC Grant  61272059,  MOE Grant 11JJD720020, NSSFC Grant 13\&ZD186, 14CZX058 and the Fundamental Research Funds for the Central Universities Grant 1409025.

\bibliography{KR2014}
\bibliographystyle{plain} 
\end{document}